\documentclass[a4paper,onecolumn,11pt,accepted=2022-03-07]{quantumarticle}
\pdfoutput=1
\usepackage[utf8]{inputenc}
\usepackage[english]{babel}
\usepackage[T1]{fontenc}
\usepackage{amsmath}
\usepackage{hyperref}

\usepackage{tikz}
\usepackage{lipsum}

\usepackage[numbers]{natbib}
\usepackage{amssymb, amsthm, amsmath, dsfont, mathtools}
\usepackage{qcircuit}

\newcommand{\replacement}[1]{#1}
\newtheorem{proposition}{Proposition}[section]
\newtheorem{lemma}[proposition]{Lemma}
\newtheorem{theorem}[proposition]{Theorem}

\newtheorem{corollary}[proposition]{Corollary}
\newtheorem{fact}[proposition]{Fact}

\newcommand{\Pn}{\mathbb{P}^{2^n-1}}

\newcommand{\Fn}{\mathbb{F}_2^n}
\newcommand{\field}{\mathbb{F}}

\newcommand{\id}{\mathds{1}}

\newcommand{\setft}[1]{\mathrm{#1}}
\newcommand{\Trans}{\setft{T}}
\newcommand{\reals}{\mathbb{R}}
\newcommand{\complex}{\mathbb{C}}
\newcommand{\naturals}{\mathbb{N}}
\DeclareMathOperator{\spn}{span}
\newcommand{\stabn}{\setft{Stab}_n}
\newcommand{\cstabn}{\check{\setft{S}}\setft{tab}_n}

\newcommand{\Unitary}{\setft{U}}

\newcommand{\aCliff}{\setft{Cliff}}
\newcommand{\proj}{\mathbb{P}}

\newcommand{\pphin}{[\phi^{\otimes n}]}

\newcommand{\psin}{\psi^{\otimes n}}
\newcommand{\taun}{\tau^{\otimes n}}
\newcommand{\ppsi}{[\psi]}
\newcommand{\pphi}{[\phi]}

\newcommand{\ceil}[1]{\lceil #1 \rceil}

\newcommand{\floor}[1]{\lfloor #1 \rfloor}

\newcommand{\tinyspace}{\mspace{1mu}}

\newcommand{\norm}[1]{\lVert\tinyspace #1 \tinyspace\rVert}
\newcommand{\bignorm}[1]{\bigl\lVert\tinyspace #1 \tinyspace\bigr\rVert}

\newcommand{\biggnorm}[1]{\biggl\lVert\tinyspace #1 \tinyspace\biggr\rVert}

\newcommand{\abs}[1]{\lvert #1 \rvert}
\newcommand{\bigabs}[1]{\bigl\lvert #1 \bigr\rvert}

\newcommand{\V}{\mathcal{V}}

\newcommand{\ket}[1]{#1}

\newcommand{\Sp}{\mathcal{S}}
\newcommand{\e}{\mathrm{e}}
\newcommand{\Spn}{\Sp(\V_n)}

\begin{document}
\mathtoolsset{showonlyrefs}

\title{New techniques for bounding stabilizer rank}

\author{Benjamin Lovitz}
\email{benjamin.lovitz@gmail.com}
\affiliation{Institute for Quantum Computing and Department of Applied Mathematics, University of Waterloo, 200 University Ave W, Waterloo, ON, Canada}
\orcid{0000-0002-2445-2701}
\author{Vincent Steffan}
\email{sv@math.ku.dk}
\affiliation{QMATH, Department of Mathematical Sciences, University of Copenhagen, Universitetsparken 5, 2100 Copenhagen, Denmark}
\orcid{0000-0003-4913-6833}

\maketitle

\begin{abstract}
In this work, we present number-theoretic and algebraic-geometric techniques for bounding the stabilizer rank of quantum states. First, we refine a number-theoretic theorem of Moulton to exhibit an explicit sequence of product states with exponential stabilizer rank but constant approximate stabilizer rank, and provide alternate (and simplified) proofs of the best-known asymptotic lower bounds on stabilizer rank and approximate stabilizer rank, up to a log factor. Second, we find the first non-trivial examples of quantum states with multiplicative stabilizer rank under the tensor product. Third,  we introduce and study the generic stabilizer rank using algebraic-geometric techniques.
\end{abstract}
 \section{Introduction}

It is of great practical importance to determine the classical simulation cost of quantum computations. Indeed, lower bounds on the simulation cost indicate quantum speedups, while upper bounds can help us to understand the limitations of quantum computation. The \textit{stabilizer rank} is a useful barometer for the computational cost of classically simulating quantum circuits under the stabilizer formalism~\cite{PhysRevX.6.021043}. A \textit{stabilizer state} is a quantum state in the orbit of a computational basis state under the action of the Clifford group (which is defined in Section~\ref{sec:basicprops}). For a (pure) quantum state $\psi$, we define its \textit{stabilizer rank}, denoted $\chi(\psi)$, to be the smallest integer $r$ for which $\psi$ can be written as a superposition of $r$ stabilizer states. The stabilizer rank is motivated by the fact that the classical simulation cost of applying Clifford gates and computational basis measurements to $\psi$ under current state-of-the-art simulation protocols scales polynomially in the number of Clifford gates and $\chi(\psi)$ (see Appendix~\ref{app:motivationofstabrank}, and~\cite{PhysRevX.6.021043,BravyiImprovedSimDominated,Bravyilowrankstabdecpmps, 2021}). For a real number $\delta >0$, the $\delta$-\textit{approximate stabilizer rank}, $\chi_{\delta}(\psi)$, is defined as the minimum stabilizer rank over all quantum states that are $\delta$-close to $\psi$, and similarly quantifies the classical simulation cost of approximating the application of Clifford gates and computational basis measurements to $\psi$ under the stabilizer formalism.

The $T$\textit{-count} of a quantum state $\psi$ is the number of $T$-gates needed to prepare $\psi$ using a circuit consisting only of Clifford+$T$-gates and post-selective computational basis measurements (we define the $T$-gate in Appendix~\ref{app:motivationofstabrank}). It is known that if $\psi$ has $T$-count $n$, then ${\chi(\psi)\leq \chi(T^{\otimes n})}$, where $T = \frac{1}{\sqrt{2}}( \mathrm{e}_0 + e^{i \pi /4} \mathrm{e}_1)$ is the so-called $T$\textit{-state}, and $\e_0, \e_1$ are the computational basis vectors in $\complex^2$~\cite{PhysRevX.6.021043}. It is therefore of particular interest to determine $\chi(T^{\otimes n})$ and $\chi_{\delta}(T^{\otimes n})$.

Despite the practical importance of the stabilizer rank, few techniques are known for bounding this quantity.
In this work, we introduce techniques from number theory and algebraic geometry for bounding the stabilizer rank. In particular, we:
\begin{enumerate}
\item Refine a theorem of Moulton on subset-sum representations of exponentially increasing sequences, and use this refinement to prove lower bounds on exact and approximate stabilizer rank. In particular, we:
\begin{itemize}
\item Prove that for any non-stabilizer qubit state $\psi$, it holds that $\linebreak{\chi(\psi^{\otimes n}) = \Omega (n/\log_2 n)}$, and in particular, ${\chi(T^{\otimes n}) \geq \frac{n+1}{4 \log_2 (n+1)}.}$ Our asymptotic scaling matches the best-known lower bound ${\chi(T^{\otimes n})\geq n/100}$ up to a log factor~\cite{peleg2021lower}.
\item Prove that for any non-stabilizer qubit state $\psi$, there exists a constant $\delta>0$ for which ${\chi_\delta(\psin) \geq \sqrt{n}/(2 \log_2 n)}$ for all $n\geq 2$. In particular, our asymptotic scaling for the $T$-state matches the best-known lower bound\linebreak ${\chi_{\delta}(T^{\otimes n})=\Omega(\sqrt{n}/\log_2 n)}$~\cite{peleg2021lower}.
\item Exhibit, for any fixed $\delta>0$, an explicit sequence of $n$-qubit product states of stabilizer rank $2^n$ (the largest possible) and $\delta$-approximate stabilizer rank $1$ (the smallest possible).
\end{itemize}
     \item Explicitly construct the first non-trivial examples of quantum states with multiplicative stabilizer rank under the tensor product.
\item Introduce and study the \textit{generic stabilizer rank}, which upper bounds $\chi(T^{\otimes n})$, using algebraic-geometric techniques.
\end{enumerate}

In the remainder of this introduction, we expand on points 1, 2, and 3, and identify directions for future work. In Section~\ref{sec:basicprops} we review some mathematical preliminaries for this work. In Sections~\ref{sec:lowerbound},~\ref{sec:multiplicative}, and~\ref{sec:genericstabilizerrank} we prove the results introduced in points 1, 2, and 3, respectively. In Appendix~\ref{app:motivationofstabrank} we review motivation for the stabilizer rank as a measure of the computational cost of classical simulations of quantum circuits.

\subsection{Lower bounds on stabilizer rank and approximate stabilizer rank}\label{sec:intro_lower}

In Section~\ref{sec:lowerbound}, we refine a number-theoretic theorem of Moulton to prove lower bounds on stabilizer rank and approximate stabilizer rank~\cite{MOULTON2001193}. Let $[n]=\{1,\dots, n\}$ when $n$ is a positive integer. For integers $q\geq 2$ and $r\geq 1$, and tuples of non-zero complex numbers
\begin{align*}
\alpha&=(\alpha_1,\dots, \alpha_q)\in \complex^q\\
\beta&=(\beta_1,\dots, \beta_r)\in \complex^r,
\end{align*}
we say that $\beta$ is a \textit{subset-sum representation} of $\alpha$ if for all $i \in [q]$ there exists a subset $R_i \subseteq [r]$ for which $\sum_{j \in R_i} \beta_j=\alpha_i$. We refer to the integer $r$ as the \textit{length} of the subset-sum representation $\beta \in \complex^r$. For an integer $2\leq p \leq q$, we say that $\alpha \in \complex^q$ has an \textit{exponentially increasing subsequence} of length $p$ if there exists $i_1,\dots, i_p \in [q]$ for which
\begin{align}\label{eq:expintro}
\abs{\alpha_{i_{j+1}}} \geq 2 \abs{\alpha_{i_j}}\quad \text{for all} \quad j \in [p].
\end{align}
Moulton's theorem states that any subset-sum representation of a $q$-tuple containing the subsequence $(1, 2, 4,\dots, 2^{p-1})$ has length at least $p/\log_2 p$~\cite{MOULTON2001193}. We refine this result to prove that the same bound holds for any $q$-tuple that contains an exponentially increasing subsequence of length $p$
(Theorem~\ref{moulton}).

Since stabilizer states have coordinates in $\{0,\pm1, \pm i\}$ in the computational basis (see Section~\ref{sec:basicprops}), any decomposition of a state $\psi$ into a superposition of $r$ stabilizer states can be converted into a length-$4r$-subset-sum representation of the coordinates of $\psi$. It follows that if the coordinates of $\psi$ contain an exponentially increasing subsequence of length $p$, then $\chi(\psi) \geq p/(4\log_2 p)$ (Theorem~\ref{thm:general_linear_lower_bound}). In particular, since $T$ is Clifford-equivalent to the \textit{$H$-state} $H \propto  \e_0+\frac{1}{\sqrt{2}-1}\e_1$, and the coordinates of $H^{\otimes n}$ contain an exponentially increasing subsequence of length $n+1$, we obtain $\chi(T^{\otimes n})\geq \frac{n+1}{4 \log_2 (n+1)}$. By a similar argument, we prove that $\chi(\psin)= \Omega (n/\log_2 n)$ for any non-stabilizer qubit state $\psi$ (Theorem~\ref{thm:general_linear_lower_bound}).

We further use Theorem~\ref{moulton}, along with standard concentration inequalities for the binomial distribution, to prove that for any non-stabilizer qubit state $\psi$ there exists a constant $\delta>0$ for which it holds that $\chi_\delta(\psin) \geq \sqrt{n}/(2\log_2 n)$ for all $n \geq 2$ (Theorem~\ref{thm:approximate stabilizer rank}).

Very recently, Peleg et al.~obtained similar results: They proved that $\chi(T^{\otimes n})\geq n/100$, and that there exists $\delta>0$ for which $\chi_{\delta}(T^{\otimes n})= \Omega (\sqrt{n}/\log_2 n)$~\cite{peleg2021lower}. Asymptotically, our bounds match theirs up to a log factor, and we suggest that our proof technique is much simpler. While both of our bounds follow quite quickly from our refinement of Moulton's theorem mentioned above, the two bounds of Peleg et al.~use two different approaches from the analysis of boolean functions and complexity theory: For their lower bound on $\chi(T^{\otimes n})$, they analyze directional derivatives of quadratic polynomials, and for their lower bound on $\chi_{\delta}(T^{\otimes n})$, they use Razborov-Smolensky low-degree polynomial approximations and correlation bounds against the majority function~\cite{zbMATH04029460,10.1145/28395.28404,10.1109/SFCS.1993.366874}. It is interesting that the vastly different approaches of ours and Peleg et al. yield such similar results.

As a further application of our refinement of Moulton's theorem, we explicitly construct a sequence of $n$-qubit product states $\psin$ for which it holds that $\chi(\psin) \geq \frac{2^n}{4 n}$ and ${\chi_{\delta}(\psin)=\mathcal{O}(1)}$ for any $\delta>0$, simply by writing down a product state with exponentially increasing coordinate amplitudes. \replacement{Using different techniques, in Proposition~\ref{prop:stabilizerrankmaximal} we construct a sequence of $n$-qubit product states $\psi^{\otimes n}$ for which $\chi(\psin)=2^n$ (the largest possible) and $\chi_{\delta}(\psin)=1$ (the smallest possible).} These results lie in contrast to the situation for other notions of rank, in which it is a difficult open problem to explicitly construct sequences of states of near-maximal rank. For example, the maximum~\textit{border rank}, a relevant notion of rank in classical complexity theory, of a quantum state in three local spaces of (affine) dimensions $d$, is $\ceil{\frac{d^3}{3d-2}}$ for all $d \neq 3$, 
whereas the largest border rank of any known explicit sequence of states in this space is only linear in $d$ (see~\cite{Landsberg2019TowardsFH} for the largest known border rank and~\cite{Landsberg2011TensorsGA} for a general introduction to the topic).

\subsection{States with multiplicative stabilizer rank under the tensor product}
It is a standard fact that the stabilizer rank is \textit{sub-multiplicative} under the tensor product, i.e. $\chi(\psi \otimes \psi) \leq \chi(\psi)^2$ for any quantum state $\psi$~\cite[Section 2.1.3]{hammam}. In~\cite[Section 4.4]{hammam} it was remarked that there are no known examples of quantum states $\psi$ of stabilizer rank greater than one for which equality holds. In Section~\ref{sec:multiplicative}, we explicitly construct two-qubit states $\psi$ for which $\chi(\psi)=2$ and $\chi(\psi \otimes \psi) = 4.$ This is the smallest possible example of such a state, since for any single-qubit state $\phi$ it holds that $\chi(\phi \otimes \phi) \leq 3$.

\subsection{Generic stabilizer rank}

For any positive integer $n$, all but finitely many qubit states $\psi$ (up to phase) maximize $\chi(\psin)$ (Fact~\ref{fact:generic}). This motivates us to define the $n$-th \textit{generic stabilizer rank}, denoted $\chi_n$, to be the maximum stabilizer rank of any state of the form $\psin$. In Section~\ref{sec:genericstabilizerrank}, we prove new bounds on $\chi_n$, along with some useful reductions for studying this quantity. In Proposition~\ref{cor:improved} we modestly improve the best-known upper bound on $\chi_n$, recently obtained by Qassim et al., from $\mathcal{O}((n+1)2^{n/2})$ to $\mathcal{O}(2^{n/2})$~\cite[Theorem 2]{2021}. In Propositions~\ref{method} and~\ref{real_reduction} we prove two useful reductions for studying $\chi_n$, namely, that there must exist a single set of $\chi_n$ stabilizer states that span the symmetric subspace, and that it suffices to work over the real numbers. In Proposition~\ref{prop:another_upper} we introduce a technique for upper bounding $\chi_n$ when upper bounds on $\chi(\psin)$ are known for sufficiently many (linear in $n$) pairwise non-collinear qubit states $\psi$. In Proposition~\ref{prop:upper_bound} we compute an upper bound on the (finite) number of qubit states $\psi$ (up to phase) for which $\chi(\psin)<\chi_n$.

\subsection{Directions for future work}
We believe that our work opens the door for new approaches on questions related to the stabilizer rank. Here, we present two of the most promising ones.

\subsubsection{Lower bounds on stabilizer rank via the $T$-count}
Recall that for any quantum state $\psi_n$ of $T$-count $n$, it holds that $\chi(T^{\otimes n})\geq \chi(\psi_n)$~\cite{PhysRevX.6.021043}. Also note that, from the proof of Theorem~\ref{thm:general_linear_lower_bound}, if the smallest subset-sum representation of $\psi_n$ has size at least $a_n \in \naturals$, then $\chi(\psi_n)\geq \frac{1}{4}a_n.$ This suggests the following technique for obtaining lower bounds on $\chi(T^{\otimes n})$ (based on a technique introduced in \cite{PhysRevX.6.021043}): Find a sequence of states $\psi_n$ of $T$-count $n$ with no subset-sum representation of size less than $a_n$, where $a_n$ grows suitably quickly in $n$. Indeed, this is the technique we use to obtain our linear lower bound in Theorem~\ref{thm:general_linear_lower_bound}: We find a sequence of states $\psi_n\propto (\e_0+\frac{i}{\sqrt{2}-1}\e_1)^{\otimes n}$ of $T$-count $n$ and no subset-sum representation of size less than $\frac{n+1}{\log_2(n+1)}$, where the latter property is ensured from the fact that $\psi_n$ contains an exponentially increasing sequence of length $n+1$ (see Theorem~\ref{moulton}). It would be interesting to see if this technique can be used to obtain super-linear lower bounds on $\chi(T^{\otimes n}).$ Unfortunately, by \cite[Proposition 5.8]{beverland2020lower}, any state of $T$-count $n$ can have an exponentially increasing sequence of length at most $O(n)$, so one would need to find another way to lower bound the size of a subset-sum representation.


\subsubsection{Bounds on generic stabilizer rank}
In Section~\ref{sec:genericstabilizerrank} we prove several reductions for studying the generic stabilizer rank $\chi_n$, but only manage to modestly improve the best-known bounds on this quantity. We ask whether stronger bounds can be obtained from our reductions.

\section*{Acknowledgments}
\replacement{We thank the second referee for the key ideas in the proof of Proposition~\ref{prop:stabilizerrankmaximal}}. We thank Gerry Myerson for pointing us to the work of Moulton~\cite{MOULTON2001193}. BL thanks Kieran Mastel and William Slofstra for helpful discussions. VS thanks Daniel Stilck Fran\c{c}a and Matthias Christandl for helpful discussions. \replacement{BL acknowledges financial support from the University of Waterloo and the Government of Ontario through an Ontario Graduate Scholarship.} VS acknowledges financial support from VILLUM FONDEN via the QMATH Centre of Excellence (Grant No. 10059).

\section{\replacement{Background}}\label{sec:basicprops}

For a complex vector space $\V$, let $\Sp(\V)$ be the set of unit vectors in $\V$ (with respect to the Euclidean norm), let $\Unitary(\V)$ be the set of unitary operators on $\V$, and let $\proj(\V)$ be the set of one-dimensional linear subspaces of $\V$. In precise terms, a pure quantum state is a unit vector modulo phase, i.e. an element of $\Sp(\V)/\Unitary(\complex)$. However, to match historical notation in quantum information, we will define a \textit{(pure, quantum) state} to be an element of $\Sp(\V)$, and define a \textit{(pure, quantum) state mod phase} to be an element of $\Sp(\V)/\Unitary(\complex)$. Under the canonical identification between this quotient and $\proj(\V)$, which sends a unit vector to its span, we also refer to elements of $\proj(\V)$ as states mod phase. This endows the set of states mod phase (and relevant subsets) with the structure of an algebraic variety. We will use this structure in Section~\ref{sec:genericstabilizerrank} to study the generic stabilizer rank.

For a vector $\psi \in \V$, let $[\psi] \in \proj(\V)$ be the subspace spanned by $\psi$ (not to be confused with $[n]:=\{1,\dots,n\}$ when $n$ is a positive integer). We say that two non-zero vectors $\psi, \phi \in \V$ are \textit{collinear} if $[\psi]=[\phi]$. For a subset $X \subseteq \proj \V$, let $\hat{X} \subseteq \V$ be the affine cone over $X$, i.e. $\hat{X}=\{\psi \in \V : [\psi] \in X\} \cup \{0\}$. For a positive integer $n$, let $\proj^n=\proj(\complex^{n+1})$ and $\V_n=(\complex^2)^{\otimes n}$. We refer to the individual copies of $\complex^2$ that make up $\V_n$ as \textit{subsystems}. We refer to elements of $\V_n$ as \textit{tensors} to emphasize the multipartite structure of $\V_n$. Let $S^n(\complex^2) \subseteq \V_n$ be the \textit{symmetric subspace}: the linear subspace of tensors that are invariant under permutation of the subsystems. Note that $\proj(\V_n)=\proj^{2^n-1}$.

For a positive integer $n$, let $\setft{Pauli}_n \subseteq \setft{U}(\V_n)$ be the \textit{Pauli group}, the group generated by all $n$-fold tensor products of elements of the set $\{ X, Z,i \id_2 \}\subseteq \setft{U}(\complex^2)$, where
\begin{equation*}
    X = \begin{pmatrix}0&1\\1&0\end{pmatrix}, \quad Z = \begin{pmatrix}1&0\\0&-1\end{pmatrix},
\end{equation*}
and $\id_k$ is the identity matrix in $\setft{U}(\complex^k)$ for any positive integer $k$. The~\textit{Clifford group}, denoted $\aCliff_n \subseteq \setft{U}(\V_n)$, is the normalizer of the Pauli group in $\setft{U}(\V_n)$. 
The Clifford group is generated by all tensor products of elements of the set $\{H, S, CNOT\}$, along with global phase gates in $\Unitary(\complex)$, where
\begin{equation*}
    H = \frac{1}{\sqrt{2}}\begin{pmatrix}1&1\\1&-1\end{pmatrix},\quad S = \begin{pmatrix}1&0\\0&i\end{pmatrix},\quad \text{and} \quad CNOT = \begin{pmatrix}1&0&0&0\\0&1&0&0\\0&0&0&1\\0&0&1&0\end{pmatrix}.
\end{equation*}

A quantum state $\psi \in \Sp(\V)$ is called a \emph{stabilizer state} if $\psi = U \mathrm{e}_0^{\otimes n}$ for some $U\in \aCliff_n$. Let $\stabn=\aCliff_n \mathrm{e}_0^{\otimes n}$ be the set of stabilizer states, and let $\cstabn=\{[\psi] : \psi \in \stabn\} \subseteq \Pn$ be the set of stabilizer states mod phase. It is well known that a state $\psi \in \Sp(\V_n)$ is a stabilizer state if and only if
\begin{equation}\label{eq:normalformofstabstate}
    \psi \propto \sum_{x \in A} i^{l(x)} \cdot (-1)^{q(x)} \cdot \mathrm{e}_x 
\end{equation}
for some affine linear subspace $A\subseteq \Fn$, linear form $l\colon \Fn \rightarrow \mathbb{F}_2$, and quadratic form $q\colon \Fn \rightarrow \mathbb{F}_2$ \cite{PhysRevA.68.042318,10.5555/2011350.2011356}. Here and throughout, we define $\e_x:=\e_{x_1} \otimes \dots \otimes \e_{x_n}$ when $x \in \Fn$.

The \emph{stabilizer rank} of a state $\psi \in \Sp(\V_n)$, denoted $\chi (\psi)$, is the smallest integer $r$ for which
\begin{equation*}
\psi = \sum_{i = 1}^r c_i \sigma_i
\end{equation*}
for some complex numbers $c_i \in \complex$ and stabilizer states $\sigma_i \in \stabn$. We denote the set of states mod phase of stabilizer rank at most $r$ by ${\Sigma}_r(\cstabn) \subseteq \Pn$ ($\Sigma_r(\cdot)$ denotes the $r$-th secant variety, see \cite[Example 8.5]{harris2013algebraic}). Note that $\cstabn$ is \textit{non-degenerate}, i.e. $\spn(\cstabn)=\Pn$.

We also have an approximate version of stabilizer rank. For a positive real number $\delta > 0$ and state $\psi \in \Sp(\V_n)$, we define the $\delta$\textit{-approximate stabilizer rank} of $\psi$, denoted $\chi_{\delta}(\psi)$, as
\begin{equation*}
\chi_\delta (\psi)=\min_{\phi \in \Sp(\V_n)}\big\{ \chi(\phi) : \bignorm{\phi- \psi} \leq \delta \big\}.
\end{equation*}

We say a quantum state $\psi\in \Sp(\V_n)$ is \textit{real} if $\psi$ is proportional to a state with only real coordinates in the computational basis. A quantum state $\psi$ is a real stabilizer state if and only if it can be written in the form~\eqref{eq:normalformofstabstate} with $l=0$. The set of real stabilizer states in $\Sp(\V_n)$, which we denote by $\stabn^\reals$, is precisely the orbit of $e_0^{\otimes n}$ under the group generated by $\Unitary(\complex)$, $H$ and $CNOT$~\cite{PhysRevA.68.042318}. We denote by $\cstabn^\reals$ the set of real stabilizer states mod phase. For a quantum state $\psi \in \Sp(\V_n)$, we define the \textit{real stabilizer rank} of $\psi$, denoted $\chi^{\reals}(\psi)$, to be the smallest integer $r$ for which $\psi$ can be written as a (complex) superposition of $r$ real stabilizer states.

We close this section by computing the number of stabilizer states and real stabilizer states mod phase, which we will use in Section~\ref{sec:genericstabilizerrank} to upper bound the number of states mod phase of sub-generic rank. It is a standard fact that there are $\binom{n}{k}_2$ distinct $k$-dimensional linear subspaces of $\field_2^n$, where
\begin{align*}
\binom{n}{k}_2=\prod_{i=0}^{k-1} \frac{2^{n-i}-1}{2^{k-i}-1}
\end{align*}
is the \textit{Gaussian binomial coefficient} (see e.g.~\cite{goldman1970foundations}). Since there are $2^{n-k}$ distinct affine translations of a $k$-dimensional linear subspace of $\Fn$, it follows that there are $\binom{n}{k}_2 2^{n-k}$ distinct affine linear subspaces of dimension $k$. For each index $k \in [n]$ and each $k$-dimensional affine subspace $A \subseteq \Fn$, there are $2^{k(k+1)/2}$ distinct quadratic forms on $A$ and $2^k$ distinct linear functions on $A$. It follows that
\begin{align*}
\abs{\cstabn}&=2^n \sum_{k=1}^n \binom{n}{k}_2 2^{k(k+1)/2}\\
			&=2^n \prod_{k=1}^n (2^k+1),
\end{align*}
where the second line follows from the Gaussian binomial theorem~(see e.g.~\cite{goldman1970foundations}). Similarly,
\begin{align*}
\bigabs{\cstabn^{\reals}}=2^n \sum_{k=1}^n \binom{n}{k}_2.
\end{align*}
The quantity $\abs{\cstabn}$ was previously computed in~\cite[Corollary~21]{gross2006hudson} using a different proof technique.

\section{Lower bounds on stabilizer rank and approximate stabilizer rank}\label{sec:lowerbound}

In this section, we refine a number-theoretic theorem of Moulton, and use this to prove lower bounds on the stabilizer rank and approximate stabilizer rank. Recall the definitions of exponentially increasing subsequences and subset-sum representations given in Section~\ref{sec:intro_lower}. Moulton proved that any subset-sum representation of a $q$-tuple containing the subsequence $(1, 2, 4, \dots, 2^{p-1})$ has length at least $p/\log_2 p$~\cite{MOULTON2001193}. In Section~\ref{sec:moulton}, we refine this result to prove that the same bound holds for any $q$-tuple containing an exponentially increasing subsequence of length $p$ (Theorem~\ref{moulton}). In Section~\ref{sec:exact} we use our refinement to prove that any quantum state whose coordinates contain an exponentially increasing subsequence of length $p$ has stabilizer rank at least $p/(4\log_2 p)$ (Theorem~\ref{thm:general_linear_lower_bound}). We then use this result to explicitly construct a sequence of product states of exponential stabilizer rank, to prove that $\chi(T^{\otimes n}) \geq \frac{n+1}{4 \log_2 (n+1)}$, and to prove that $\chi(\psin) = \Omega(n/\log_2 n)$ for any non-stabilizer qubit state $\psi$. In Section~\ref{sec:approx} we use our refinement of Moulton's theorem to prove that, for any non-stabilizer qubit state $\psi$, there exists $\delta>0$ for which $\chi_{\delta}(\psi^{\otimes n})\geq \sqrt{n}/(2\log_2 n)$ for all $n \in \mathbb{N}$. \replacement{In Section~\ref{sec:exp} we use our refinement of Moulton's theorem, as well as independent field-theoretic techniques, to construct explicit sequences of product states of exponential stabilizer rank but constant approximate stabilizer rank.}

\subsection{A refinement of Moulton's theorem}\label{sec:moulton}
\begin{theorem}[Refinement of Theorem~1 in~\cite{MOULTON2001193}]\label{moulton}
Let $2 \leq p \leq q$ be integers, and let $\alpha \in \complex^q$ be a $q$-tuple of non-zero complex numbers. If $\alpha$ contains an exponentially increasing subsequence of length $p$, then any subset-sum representation of $\alpha$ has length at least $p/\log_2(p)$.
\end{theorem}
\begin{proof}
It suffices to consider the case $p=q$ and $2\abs{\alpha_i}\leq\abs{\alpha_{i+1}}$ for all $i \in [q-1]$. Let $\beta \in \complex^r$ be a subset-sum representation of $\alpha$. Then for each $i \in [q]$, there exists $c_i \in \{0,1\}^r$ such that $\alpha_i=\beta^{\Trans} c_i$ Suppose that, for some $u_1,\dots, u_q, v_1,\dots, v_q \in \{0,1\}$, we have
\begin{align*}
\sum_{i=1}^q u_i c_i=\sum_{i=1}^q v_i c_i.
\end{align*}
Applying $\beta^{\Trans}$ to both sides gives
\begin{align*}
\sum_{i=1}^q u_i \alpha_i=\sum_{i=1}^q v_i \alpha_i.
\end{align*}
It follows that $u_i=v_i$ for all $i\in [q]$. Indeed, it suffices to prove that $\abs{\alpha_{i+1}} > \abs{\alpha_1+\dots + \alpha_i}$ for all $i \in [q-1]$, which in turn can be easily verified by an inductive argument. By assumption, $\abs{\alpha_{2}}>\abs{\alpha_1}$, and by induction,
\begin{align*}
\abs{\alpha_1+\dots+\alpha_i} \leq \abs{\alpha_1+\dots + \alpha_{i-1}}+\abs{\alpha_i} < 2\abs{\alpha_i} \leq \abs{\alpha_{i+1}}.
\end{align*}

The remainder of the proof is identical to that of \cite{MOULTON2001193}. There are at most ${2^q-1}$ choices of $u_1,\dots, u_q \in \{0,1\}$, excluding the case $u_1=\dots=u_q=1$. For each of these choices, the sum $\sum_{i=1}^q u_i c_i$ can take one of $q^r-1$ possible choices in $\{0,1,\dots, q-1\}^{\times r}$ (the choice $(q-1,q-1,\dots, q-1)^\Trans$ is excluded since the $u_i$ are not all equal to 1). Since each choice of $u_1,\dots, u_q$ yields a different vector, we must have $ q^r-1 \geq 2^q-1$, i.e. ${r \geq q/\log_2(q)}$.
\end{proof}

\subsection{Lower bounds on stabilizer rank}\label{sec:exact}

In this subsection we use Theorem~\ref{moulton} to prove lower bounds on stabilizer rank.

\begin{theorem}\label{thm:general_linear_lower_bound}
Let $p \geq 2$ be an integer, and let $\psi \in \Spn$ be a quantum state. If the coordinates of $\psi$ contain an exponentially increasing subsequence of length $p$, then $\chi(\psi) \geq p/ (4\log_2 p)$.
\end{theorem}
\begin{proof}
Let $r=\chi(\psi)$, let $x_1,\dots, x_p \in \Fn$ be such that $\abs{\psi_{x_i}}\leq 2 \abs{\psi_{x_{i+1}}}$ for all ${i \in [p-1]}$, and let $\alpha=(\psi_{x_1},\dots, \psi_{x_p})\in \complex^p$. Without loss of generality, there exist complex numbers $\{c_i : i \in [r]\}\subseteq \complex$ and stabilizer states $\{\sigma_i : i \in [r]\}\subseteq \stabn$ such that for all $i \in [r]$, every coordinate of $\sigma_i$ is an element of $\{0,\pm1, \pm i\},$ and $\psi=\sum_{i=1}^r c_i \sigma_i$. Let
\begin{align}
S=(\sigma_1,\dots, \sigma_r) \in \{0,\pm 1, \pm i\}^{\{0,1\}^n \times r}
\end{align}
and
\begin{align}
{c=(c_1,\dots, c_r) \in \complex^{r}},
\end{align}
so that $S c=\psi$. In particular, there exists a $p \times r$ submatrix $T$ of $S$ for which $Tc=\alpha$.
Let $T_1, T_2, T_3, T_4 \in \{0,1\}^{p \times r}$ be such that 
$${T=T_1-T_2+i(T_3-T_4)}.$$
Then
\begin{align}
(T_1,T_2,T_3,T_4)(c,-c,ic,-ic)^\Trans=Tc=\alpha,
\end{align}
so $(c,-c,ic,-ic)$ is a subset-sum representation of $\alpha$. It follows from Theorem~\ref{moulton} that $4r \geq p/(\log_2 p)$. This completes the proof.
\end{proof}

Theorem~\ref{thm:general_linear_lower_bound} also implies the following lower bound on $\chi(T^{\otimes n})$, and more generally, on $\chi(\psin)$ for any non-stabilizer qubit state $\psi$.
\begin{corollary}\label{linear_lower_bound}
For any state $\psi \in \Sp(\complex^2)$ that is not a stabilizer state, $\chi(\psin) = \Omega(n/\log_2 n)$. In particular,
\begin{align}\label{eq:T_lower_bound}
\chi(T^{\otimes n}) \geq \frac{n+1}{4 \log_2 (n+1)}.
\end{align}
\end{corollary}

\begin{proof}
Since $\psi$ is not a stabilizer state, there exists $\alpha \in \complex$ with $\abs{\alpha}>1$ for which ${\tau:= \frac{1}{\sqrt{1+\abs{\alpha}^2}} (\mathrm{e}_0+\alpha \mathrm{e}_1})$ is in the $\aCliff_1$-orbit of $\psi$. Indeed, by adjusting the global phase of $\psi$ we may assume $\psi = \frac{1}{\sqrt{1+\abs{\beta}^2}} (\e_0 + \beta \e_1)$ for some complex number $\beta \in \complex^{\times}\setminus \{0,\pm 1, \pm i\}$, because $\psi$ is not a stabilizer state. If $\abs{\beta}>1$ we are done, and if $\abs{\beta}<1$ then let $\tau = X \psi$. If $\abs{\beta}=1$ then either $H \psi$ or $XH \psi$ must have the desired form. When $\psi=T$, we can take $\tau = XH\psi \propto \mathrm{e}_0+\frac{i}{\sqrt{2}-1}\mathrm{e}_1$.

Since $\abs{\alpha} >1$, there exists $k \in \naturals$ for which $\abs{\alpha}^k \geq 2$. (When $\psi=T$, we can take ${k=1}$.) Now observe that the complex numbers $1,\alpha^k, \alpha^{2k},\dots, \alpha^{\floor{n/k} k}$ all appear as coordinates of $\psi^{\otimes n}$. By Theorem~\ref{thm:general_linear_lower_bound}, it follows that
\begin{align}\label{bound}
\chi(\psi^{\otimes n})\geq \frac{\floor{n/k}+1}{4\log_2 (\floor{n/k}+1)}.
\end{align}
This completes the proof.
\end{proof}

\subsection{Lower bounds on approximate stabilizer rank}\label{sec:approx}

In this subsection, we use Theorem~\ref{moulton} to prove that, for any non-stabilizer qubit state $\psi$, there exists a constant $\delta>0$ for which $\chi_\delta (\psin)\geq \sqrt{n}/(2\log_2n).$

\begin{theorem}\label{thm:approximate stabilizer rank}
For any non-stabilizer qubit state $\psi \in \Sp(\complex^2)$, there exists a constant $\delta > 0$ such that, for every integer $n \geq 2$,
\begin{equation*}
\chi_\delta ( \psi^{\otimes n}) \geq \frac{{\sqrt{n}}}{2 \log_2 n}.
\end{equation*}

\end{theorem}
\begin{proof}

As in the proof of Corollary~\ref{linear_lower_bound}, since $\psi$ is not a stabilizer state, there exists $\alpha \in \complex$ with $\abs{\alpha}>1$ for which $\tau:=\frac{1}{\sqrt{1+\abs{\alpha}^2}}(\mathrm{e}_0+\alpha \mathrm{e}_1) \in \aCliff_1 (\psi)$.  Let $\beta=\frac{1}{\sqrt{1+\abs{\alpha}^2}} $ and $\gamma=\frac{\alpha}{\sqrt{1+\abs{\alpha}^2}},$ so ${\tau=\beta \mathrm{e}_0+\gamma \mathrm{e}_1}$. Since the (approximate) stabilizer rank is unchanged under $\aCliff_n$, it suffices to lower bound the approximate stabilizer rank of $\tau^{\otimes n}$.

Let $k \in \naturals$ be the smallest integer for which $\abs{\alpha}^k > 2$, and let $\lambda=\frac{2}{\abs{\alpha}^k}$. If $\psi=T$, then we can take $\alpha=\frac{1}{\sqrt{2}-1}$ and $k=1$. Let $\phi \in \Sp(\V_n)$ be a state, and let
$S\subseteq [n]$ be the set of integers $p \in [n]$ \replacement{that satisfy the following two properties:}
\begin{enumerate}
\item $\abs{\gamma}^2 n- k \ceil{\sqrt{n}}\leq  p \leq \abs{\gamma}^2 n + k \ceil{\sqrt{n}}.$
\item For all $x \in \Fn$ of Hamming weight $\abs{x}=p$,
\begin{align}\label{eq:S_inequality}
\abs{\tau_x^{\otimes n} - \phi_x} \geq \left(\frac{1-\lambda}{1+\lambda}\right) \abs{\beta}^{n-p} \abs{\gamma}^{p}.
\end{align}
\end{enumerate}
\replacement{By the De Moivre-Laplace Theorem, there exists a constant $\tilde{c} >0$ (which may depend on $\abs{\alpha}$, but does not depend on $n$) for which}
\begin{align}
 \binom{n}{p} \abs{{\beta}^{n-p} {\gamma}^{p}}^2 \geq \tilde{c}/\sqrt{n}
\end{align}
\replacement{for all $p \in [n]$ that satisfy property 1 (see \cite[Section VII, Theorem 1]{feller} or \cite[Claim 4.6]{peleg2021lower}). Let}
\begin{align}
c=\left(\frac{1-\lambda}{1+\lambda}\right)^2 \tilde{c},
\end{align}
\replacement{so that}
\begin{align}
\left(\frac{1-\lambda}{1+\lambda}\right)^2 \binom{n}{p} \abs{{\beta}^{n-p} {\gamma}^{p}}^2 \geq c/\sqrt{n}
\end{align}
for all $p \in S$. It follows that
\begin{align}
\bignorm{\tau^{\otimes n} -\phi}^2 \geq \sum_{p \in S} \left(\frac{1-\lambda}{1+\lambda}\right)^2 \binom{n}{p} \abs{\beta^{n-p} \gamma^{p}}^2\geq \abs{S} \frac{c}{\sqrt{n}}.
\end{align}
Let $\delta=\sqrt{ck}$, and suppose that $\bignorm{\tau^{\otimes n} -\phi}\leq \delta$. Then $\abs{S} \leq k \sqrt{n}$. Let $P\subseteq [n]$ be the set of integers $p\in [n]$ that satisfy property 1 but do not satisfy property 2. Observe that $\abs{P}\geq k \ceil{\sqrt{n}}$, \replacement{because there are at least $2k \ceil{\sqrt{n}}$ integers that satisfy property 1, and $\abs{S}\leq k \sqrt{n}$.} By our definition of $P$, for each $p \in P$ there exists $x_p \in \Fn$ with $\abs{x_p}=p$ and
\begin{align}\label{eq:Pbound}
\bigabs{\tau^{\otimes n}_{x_p} - \phi_{x_p}} \leq \left(\frac{1-\lambda}{1+\lambda}\right) \abs{\beta}^{n-p} \abs{\gamma}^{p}.
\end{align}
For any $p,q \in P$ with $p<q$, it holds that
\begin{align}
\frac{\abs{\phi_{x_{q}}}}{\bigabs{\phi_{{x_{p}}}}}&=\frac{\bigabs{\phi_{x_q}-\tau^{\otimes n}_{x_q}+\taun_{x_q}}}{\bigabs{\phi_{x_p}-\tau^{\otimes n}_{x_p}+\taun_{x_p}}}\\
&\geq \frac{\bigabs{\taun_{x_q}}-\bigabs{\phi_{x_q}-\tau^{\otimes n}_{x_q}}}{\bigabs{\taun_{x_p}}+\bigabs{\phi_{x_p}-\tau^{\otimes n}_{x_p}}}\\
&\geq \frac{(1-\frac{1-\lambda}{1+\lambda}) \abs{\beta^{n-q} \gamma^q}}{(1+\frac{1-\lambda}{1+\lambda}) \abs{\beta^{n-p} \gamma^p}}\\
&=\lambda \abs{\alpha}^{q-p},
\end{align}
where the first line is trivial, the second is the triangle inequality, \replacement{the third follows from~\eqref{eq:Pbound} and the fact that $\bigabs{\taun_{x}}=\abs{\beta^{n-\abs{x}}\gamma^{\abs{x}}}$ for all $x \in \field_2^n$}, and the fifth is obvious. In particular, if $q-p \geq k$, then $\frac{\abs{\phi_{x_{q}}}}{\abs{\phi_{{x_{p}}}}}\geq 2$. Since $\abs{P}\geq k \ceil{\sqrt{n}}$, there exists a subset $Q\subseteq P$ of size $\abs{Q} \geq \sqrt{n}$ for which $p-q \geq k$ for all $p,q \in Q$ with $p<q$.
By Theorem~\ref{thm:general_linear_lower_bound},
\begin{align}
\chi(\phi) \geq \frac{\sqrt{n}}{2 \log_2(n)}.
\end{align}
This completes the proof.
\end{proof}

We note that, using~\cite[Section VII, Theorem 1]{feller} and a similar proof as above, slight improvements to the above bound can be obtained if $\delta$ is allowed to decay in $n$.

\subsection{\replacement{Product states with exponential stabilizer rank and constant approximate stabilizer rank}}\label{sec:exp}

It is not difficult to prove that for any $\delta>0$, there exist product states $\psi_n \in \Spn$ with $\chi(\psi_n)=2^n$ (the largest possible) and ${\chi_\delta(\psi_n)=1}$ (the smallest possible). Indeed, generic product states have stabilizer rank $2^n$ (see the discussion following Proposition~\ref{method}), so a $\delta$-ball around any product stabilizer state (e.g. $\e_0^{\otimes n}$) will contain (many) product states of stabilizer rank $2^n$. In this subsection, we will explicitly construct examples of such product states. In Corollary~\ref{lemma:constant_approx} we use Theorem~\ref{thm:general_linear_lower_bound} to provide a simple proof that a particular sequence of product states has stabilizer rank at least $\frac{2^n}{4n}$ and $\delta$-approximate stabilizer rank $\mathcal{O}(1)$. In Proposition~\ref{prop:stabilizerrankmaximal}, we use independent, field-theoretic techniques to construct a different sequence of product states of stabilizer rank $2^n$ and $\delta$-approximate stabilizer rank 1. The following lemma will allow us to upper-bound the $\delta$-approximate stabilizer rank of product states with exponentially increasing coordinates.

\begin{lemma}\label{lemma:constant_approx}
Let $\delta>0$ be a positive real number. For a complex number $\theta \in \complex$ and natural number $n \in \naturals$, let
\begin{equation*}
    \psi^\theta_n =\sqrt{\frac{\abs{\theta}^2 - 1}{\abs{\theta}^{2^{n+1}}-1}} \bigotimes_{i=1}^n (\e_0+\theta^{2^{i-1}} \e_1)\in \Spn.
\end{equation*}
If $\abs{\theta}>1$, then $\chi_{\delta}(\psi_n^\theta) = \mathcal{O}(1)$. Furthermore, there exists a positive real number $\lambda_\delta >0$ such that for any $\theta \in \complex$ with $\abs{\theta} > \lambda_\delta$, it holds that $\chi_\delta (\psi_n^\theta) = 1$ for all $n \in \mathbb{N}$.
\end{lemma}
\begin{proof}
For each $i \in \{0,1,\dots, 2^n-1\}$, let
\begin{align}
c_i=\sum_{j=0}^i \abs{\theta}^{2j}=\frac{{\abs{\theta}^{2i+2}-1}}{\abs{\theta}^2 - 1},
\end{align}
and observe that for any positive integer $k$, the tensor $\phi_{n,k}^\theta \in \V_n$ obtained by setting all but the $k$ largest coordinates of $\psi_n^\theta$ to zero satisfies

\begin{align}
\biggnorm{{\psi_n^\theta} - \frac{\phi_{n,k}^\theta}{\norm{\phi_{n,k}^\theta}}}^2
&=\biggnorm{\frac{1}{\sqrt{c_{2^n - 1}}}\sum_{i=0}^{2^n -1} \theta^i \e_i-\frac{1}{\sqrt{c_{2^n - 1}-c_{2^n-k -1}}}\sum_{i=2^n-k}^{2^n -1} \theta^i \e_i}^2\\
&=\frac{c_{2^n-k-1}}{c_{2^n-1}}+\left(\frac{1}{\sqrt{c_{2^n-1}}}-\frac{1}{\sqrt{c_{2^n-1}-c_{2^n-k-1}}}\right)^2 (c_{2^n-1}-c_{2^n-k-1})\\
&=\frac{c_{2^n-k-1}}{c_{2^n-1}}+\left[\sqrt{1-\frac{c_{2^n-k-1}}{c_{2^n-1}}}-1\right]^2,\\
&=\frac{\abs{\theta}^{-2k}-\abs{\theta}^{-2^{n+1}}}{1-\abs{\theta}^{-2^{n+1}}}+\left[\sqrt{\frac{1-\abs{\theta}^{-2k}}{1-\abs{\theta}^{-2^{n+1}}}}-1\right]^2,\label{eq:constant}
\end{align}
where we have re-indexed the computational basis of $\V_n$ as $\e_0,\dots, \e_{2^n-1}$ for clarity in this proof. Since $\abs{\theta}>1$, the quantity~\eqref{eq:constant} can be set to less than $\delta^2$ by appropriate choice of $k=\mathcal{O}(1)$. Since $\chi(\phi_{n,k}^\theta)\leq k$, this shows that $\chi_\delta (\psi_n^\theta) = \mathcal{O}(1)$. It is clear that we can set $k=1$ if $\abs{\theta}$ is large enough. This completes the proof.
\end{proof}

Note that, using Theorem~\ref{thm:general_linear_lower_bound} and Lemma~\ref{lemma:constant_approx}, we can easily find a sequence of product states of stabilizer rank at least $\frac{2^n}{4n}$ and $\delta$-approximate stabilizer rank $\mathcal{O}(1)$.
\begin{corollary}\label{rmk:explicitlargerank}
For any $n \in \naturals$, let
\begin{align}\label{eq:exp}
\psi_n =\sqrt{\frac{3}{4^{2^n}-1}} \bigotimes_{i=1}^n (\e_0+2^{2^{i-1}} \e_1)\in \Spn
\end{align}
be a quantum state. Then $\chi(\psi_n) \geq \frac{2^n}{4n}$ and for any constant $\delta>0$, $\chi_{\delta}(\psi_n) = \mathcal{O}(1)$.
\end{corollary}
\begin{proof}
For each $n$, the coordinates of $\psi_n$ form an exponentially increasing sequence of length $2^n$. It follows from Theorem~\ref{thm:general_linear_lower_bound} that $\chi(\psi_n) \geq \frac{2^n}{4n}$. The bound $\chi_{\delta}(\psi_n) = \mathcal{O}(1)$ follows from Lemma~\ref{lemma:constant_approx}.
\end{proof}

We now use field-theoretic techniques to construct a sequence of product states of stabilizer rank $2^n$ and $\delta$-approximate stabilizer rank $1$.

\begin{proposition}\label{prop:stabilizerrankmaximal}
Let $\delta>0$ be a positive real number. If $\theta \in \mathbb{C}$ is a complex number of degree at least $2^n$ over $\mathbb{Q}(i)$, then the state 
\begin{equation*}
    \psi^\theta_n =\sqrt{\frac{\abs{\theta}^2 - 1}{\abs{\theta}^{2^{n+1}}-1}} \bigotimes_{i=1}^n (\e_0+\theta^{2^{i-1}} \e_1)\in \Spn
\end{equation*}
has stabilizer rank $\chi(\psi_n^\theta) = 2^n$. If it furthermore holds that $\abs{\theta}>1$, then $\chi_{\delta}(\psi_n^\theta) = \mathcal{O}(1)$. Finally, there exists a positive real number $\lambda_{\delta}$ such that for every $\theta \in \mathbb{C}$ of degree at least $2^n$ over $\mathbb{Q}(i)$ for which $\abs{\theta}>\lambda_{\delta}$, it holds that $\chi(\psi_n^\theta)=2^n$ and $\chi_\delta (\psi_n^\theta) = 1$.
\end{proposition}
For example, to obtain a product state with stabilizer rank $2^n$ and approximate stabilizer rank $\mathcal{O}(1)$, one can choose any transcendental number $\theta \in \complex$, e.g. the circumference-to-diameter ratio $\pi.$ Note also that there are $\theta \in \mathbb{C}$ with $|\theta|=1$ such that Proposition~\ref{prop:stabilizerrankmaximal} applies. It is worth mentioning that in this case, the coordinates of $\psi_n^\theta$ contain no exponentially increasing sequence, that is, Theorem~\ref{thm:general_linear_lower_bound} does not yield an immediate lower bound on $\chi (\psi_n^\theta )$.
\begin{proof}
Note that the coordinates of $\psi_n^\theta$ are (up to the normalization factor) $1, \theta, \dots , \theta^{2^n-1}$. By the degree assumption on $\theta$, the coordinates of $\psi_n^\theta$ are linearly independent over $\mathbb{Q}(i)$.
Let $\psi_n^\theta=\sum_{i=1}^r c_i \sigma_i$ be a stabilizer decomposition of $\psi_n^\theta$, where $\sigma_i\in\stabn$ and $c_i\in \complex$ for each $i \in [r]$. Since the coordinates of each $\sigma_i$ are contained in $\mathbb{Q}(i)$, it follows that
\begin{align}
\spn_{\mathbb{Q}(i)}\{c_1,\dots, c_r\} \supseteq \{1, \theta, \dots, \theta^{2^n-1}\},
\end{align}
so $r \geq 2^{n}$. This completes the proof that $\chi(\psi_n^\theta) = 2^n$. The remaining statements follow from Lemma~\ref{lemma:constant_approx}.
\end{proof}

\section{States with multiplicative stabilizer rank under the tensor product}\label{sec:multiplicative}
It is a standard fact that the stabilizer rank is \textit{sub-multiplicative} under the tensor product, i.e. $\chi(\psi \otimes \psi) \leq \chi(\psi)^2$ for any quantum state $\psi$~\cite[Section 2.1.3]{hammam}. In~\cite[Section 4.4]{hammam} it was remarked that there are no known examples of quantum states $\psi$ of stabilizer rank greater than one for which equality holds. In this section, we explicitly construct two-qubit states $\psi$ for which $\chi(\psi)=2$ and $\chi(\psi \otimes \psi) = 4$. Note that this is the smallest non-trivial example of multiplicative stabilizer rank that one can hope for, since for any single-qubit state $\psi\propto \mathrm{e}_0+ \alpha \mathrm{e}_1$, it holds that
$\psi^{\otimes 2}\propto \e_{00}+\alpha(\e_{01}+\e_{10}) + \alpha^2 \e_{11}$, which has stabilizer rank at most 3.\\
 
\begin{theorem}
Let
\begin{align}
\psi_\alpha=\frac{1}{\sqrt{1+2\abs{\alpha}^2}} (\e_{00}+\alpha(\e_{01}+\e_{10})) \in \Sp((\complex^2)^{\otimes 2})
\end{align}
when $\alpha \in \complex^ \times$ is a non-zero complex number. Then $\chi(\psi_\alpha)=2$, and for all but finitely many $\alpha$ it holds that $\chi(\psi_\alpha^{\otimes 2})=4$. In particular, $\chi(\psi_\alpha^{\otimes 2})=4$ for any $\alpha$ that is transcendental over $\mathbb{Q}$.
\end{theorem}
\begin{proof}
The fact that $\chi(\psi_\alpha)=2$ is obvious, so it suffices to prove that $\chi(\psi_\alpha^{\otimes 2})=4$. Note that, since the imaginary unit $i$ is algebraic over $\mathbb{Q}$, $\alpha$ is transcendental over $\mathbb{Q}$ if and only if it is transcendental over $\mathbb{Q}(i)$. Since there are only finitely many stabilizer states mod phase, it suffices to prove that for any set of three stabilizer states $\{\sigma_1,\sigma_2,\sigma_3\} \subseteq \setft{Stab}_2$, there are at most finitely many $\alpha$ (and no $\alpha$ which are transcendental over $\mathbb{Q}$) for which
\begin{align}\label{eq:alpha}
\psi_\alpha^{\otimes 2}\in \spn \{\sigma_1,\sigma_2,\sigma_3\}.
\end{align}
For each $j\in [3],$ we may assume that
\begin{align*}
\sigma_j = \frac{1}{\sqrt{\abs{A_j}}} \sum_{x \in A_j} i^{l_j(x)} \cdot (-1)^{q_j(x)} \cdot \mathrm{e}_x
\end{align*}
for some affine linear subspace $A_j \subseteq \field_2^2$, linear functional $l_j$, and quadratic form $q_j$, as in~\eqref{eq:normalformofstabstate}. Let $S=(\sigma_1,\sigma_2,\sigma_3)\in \mathbb{Q}(i)^{\{0,1\}^{\replacement{2}} \times 3}$. For every $\alpha$ that satisfies~\eqref{eq:alpha}, there exist complex numbers $\beta_\alpha, \gamma_\alpha, \lambda_\alpha \in \complex$ for which $\psi_\alpha^{\otimes 2}=\beta_\alpha \sigma_1+\gamma_\alpha \sigma_2+\lambda_\alpha \sigma_3$, and a $3 \times 3$ submatrix $T$ of $S$ for which 
\begin{align}\label{eq:T}
\begin{bmatrix}
T &\replacement{\vert} -\id_3
\end{bmatrix}
\begin{bmatrix}
\beta_\alpha\\
\gamma_\alpha\\
\lambda_\alpha\\
1\\
\alpha\\
\alpha^2
\end{bmatrix}=0.
\end{align}
Since there are only finitely many $3 \times 3$ submatrices of $S$, it suffices to prove that for any choice of $3 \times 3$ submatrix $T$, only finitely many $\alpha$ (and no transcendental $\alpha$) can satisfy~\eqref{eq:T}.

If $T$ is singular, then there exists a non-trivial $\mathbb{Q}(i)$-linear combination of $\{1,\alpha, \alpha^2\}$ that equals zero, which at most finitely many $\alpha$ will satisfy (and in particular, no transcendental $\alpha$ will satisfy).

If $T$ is nonsingular, then by applying $T^{-1}$ to both sides we can find polynomials ${f,g,h \in \mathbb{Q}(i)[x]}$ for which $\beta_\alpha=f(\alpha), \gamma_\alpha=g(\alpha)$, and $\lambda_\alpha=h(\alpha)$ for every choice of $\alpha$ that satisfies~\eqref{eq:T}.

Let $R=\{00,01,10\}^{\times 2} \subseteq \mathbb{F}_2^4$ be the subset of bitstrings corresponding to non-zero coordinates of $\psi_\alpha^{\otimes 2}$. If there exists $x \in A_i \setminus R$ for some $i \in [3]$, then there exists a non-trivial $\mathbb{Q}(i)$-linear combination of $\{f,g,h\}$ which equals zero, which at most finitely many $\alpha$ (and no transcendental $\alpha$) will satisfy. So we may assume $A_i \subseteq R$ for all $i \in [3]$.

Let $R_0=\{0000\}$, $R_1=\{0001,0010,0100,1000\}$, and $R_2=\{0101, 0110,1001,1010\}$, so $R=R_0\cup R_1 \cup R_2$. If there exists $x \in A_1 \cap R_2$ and $y \in R_2 \setminus A_1$, then there exist linear combinations
\begin{align*}
\star f(\alpha)+&\cdot g(\alpha) + \cdot h(\alpha)=\alpha^2\\
&\cdot g(\alpha) + \cdot h(\alpha) = \alpha^2,
\end{align*}
where each $\cdot \in \complex$ denotes an arbitrary complex number we do not bother to name, and $\star \in \complex^\times$ denotes an arbitrary non-zero complex number we do not bother to name. Subtracting these equations yields a non-trivial $\mathbb{Q}(i)$-linear combination of $\{f,g,h\}$ that equals zero, which at most finitely many $\alpha$ (and no transcendental $\alpha$) will satisfy. By symmetry, for all $i \in [3]$, either $R_2 \subseteq A_i$ or $R_2 \cap A_i = \emptyset$. If $R_2 \subseteq A_i$, then $R_2=A_i$ because this is the only affine linear subspace contained in $R$ that contains $R_2$. Without loss of generality, we may assume $A_1=R_2$. If $A_2=R_2$, then $A_3=R_0 \cup \replacement{R_1}$, which is not an affine linear subspace, a contradiction. Otherwise, $A_2 \cup A_3 = R_0 \cup R_1$, which is also a contradiction, as there are no affine linear subspaces of size greater than 2 contained in $R_0 \cup R_1$. This completes the proof.
\end{proof}

\section{Generic stabilizer rank}\label{sec:genericstabilizerrank}
\begin{fact}\label{fact:generic}
For any positive integer $n$, all but finitely many qubit states mod phase $\ppsi \in \proj^1$ maximize $\chi(\psin)$. Similarly, all but finitely many qubit states mod phase $\ppsi \in \proj^1$ maximize $\chi^{\reals}(\psin)$.
\end{fact}
This fact, the proof of which we defer until after introducing this section, motivates us to define the $n$-th \textit{generic stabilizer rank} as
\begin{align}
\chi_n=\max_{\psi \in \Sp(\complex^2)} \chi(\psi^{\otimes n}),
\end{align}
and the $n$-th \textit{generic real stabilizer rank} as
\begin{align}
\chi_n^{\reals}=\max_{\psi \in \Sp(\complex^2)} \chi^{\reals}(\psi^{\otimes n}).
\end{align}
In Proposition~\ref{real_reduction} we prove that $\chi_n$ and $\chi_n^{\reals}$ differ by at most a constant factor of two. This motivates the study of $\chi_n^{\reals}$, since it has the same scaling as $\chi_n$, but may be easier to work with. Note that $\chi_n$ upper bounds $\chi(T^{\otimes n})$.

In this section, we prove bounds on $\chi_n$ and $\chi_n^{\reals}$, and observe some useful reductions for studying these quantities. In Proposition~\ref{cor:improved} we use Fact~\ref{fact:generic} to obtain a modest improvement of the best-known upper bound on $\chi_n$. In Proposition~\ref{method} we prove a useful reduction for studying $\chi_n$ (respectively, $\chi_n^{\reals}$), namely, that there must exist a single set of $\chi_n$ stabilizer states (resp. $\chi_n^{\reals}$ real stabilizer states) that span the symmetric subspace. In other words, there exists a single set of $\chi_n$ stabilizer states (resp. $\chi_n^{\reals}$ real stabilizer states) such that, for every qubit state $\psi$, the state $\psin$ lies in the span of the set. In Proposition~\ref{prop:another_upper} we introduce a technique for upper bounding $\chi_n$ when upper bounds on $\chi(\psin)$ are known for sufficiently many (linear in $n$) qubit states mod phase $\ppsi \in \proj^1$, and obtain similar results for $\chi_n^{\reals}$. In Proposition~\ref{prop:upper_bound} we refine Fact~\ref{fact:generic} to obtain quantitative upper bounds on the (finite) number of states mod phase of sub-generic stabilizer rank and sub-generic real stabilizer rank.

Before proceeding, we prove Fact~\ref{fact:generic}. We refer the reader to~\cite{harris2013algebraic} for the basic algebraic-geometric definitions and arguments used in this proof, and elsewhere in this section. In this proof and in others in this section, we only prove the statement for $\chi_n$, as the proof for $\chi_n^{\reals}$ is essentially identical.
\begin{proof}[Proof of Fact~\ref{fact:generic}]
We prove only the statement for $\chi_n$. Let
\begin{align}
\nu_n(\proj^1)=\{\pphin : \pphi \in  \proj^1\}\subseteq \Pn,
\end{align}
which happens to be the image of the $n$-th \emph{Veronese embedding} of $\proj^1$, embedded in $\Pn$. The set $\nu_n(\proj^1)$ forms a 1-dimensional irreducible projective variety, so its intersection with any other projective variety is either empty, zero-dimensional, or equal to $\nu_n(\proj^1)$. Hence, for any positive integer $r$, $\Sigma_r(\cstabn) \cap \nu_n(\proj^1)$ is either a finite set of points, or $\Sigma_r(\cstabn) \supseteq \nu_n(\proj^1).$ This completes the proof.
\end{proof}

Fact~\ref{fact:generic} immediately implies the following slight improvement of the upper bound $\chi_n= \mathcal{O}((n+1)2^{n/2})$ obtained in~\cite[Theorem 3]{2021}:

\begin{proposition}\label{cor:improved}
$\chi_n = \mathcal{O}(2^{n/2})$.
\end{proposition}
\begin{proof}
This follows directly from our Fact~\ref{fact:generic} and~\cite[Theorem 2]{2021}, which states that equitorial states (an infinite family of states mod phase) have stabilizer rank $\mathcal{O}(2^{n/2})$.
\end{proof}

We next prove that a single set of $\chi_n$ stabilizer states (respectively, $\chi_n^{\reals}$ real stabilizer states) spans $S^n(\complex^2).$

\begin{proposition}\label{method}
For any positive integer $n$, there exists a set of stabilizer states \linebreak$\{\sigma_1,\dots, \sigma_{\chi_n}\}\subseteq \stabn$ for which
\begin{align}
S^n(\complex^2)\subseteq \spn\{\sigma_1,\dots, \sigma_{\chi_n}\}.
\end{align}
Similarly, there exists a single set of real stabilizer states $\{\sigma_1,\dots, \sigma_{\chi_n^{\reals}}\}\subseteq \stabn^{\reals}$ for which
\begin{align}
S^n(\complex^2)\subseteq \spn\{\sigma_1,\dots, \sigma_{\chi_n^{\reals}}\}.
\end{align}
\end{proposition}
Since $\dim(S^n(\complex^2))=n+1$, it follows from Proposition~\ref{method} that $\chi_n \geq n+1$. We note that a similar proof as below can be used to show that a generic product state of the form $\psi_1 \otimes \dots \otimes \psi_n$, for qubit states $\psi_i \in \Sp(\complex^2)$, has stabilizer rank $2^n$, where in this context we define \textit{generic} in the same algebraic-geometric sense as in~\cite{entangled_sub}. \replacement{In Proposition~\ref{prop:stabilizerrankmaximal} we have presented an explicit sequence of product states of stabilizer rank $2^n$, which is maximal.}
\begin{proof}[Proof of Proposition~\ref{method}]
We prove only the statement for $\chi_n$. Let $r=\chi_n$, so $\Sigma_r(\cstabn) \supseteq \nu_n(\proj^1),$ where $\nu_n(\proj^1)$ is defined in the proof of Fact~\ref{fact:generic}. Since $\nu_n(\proj^1)$ is irreducible, and $\Sigma_r(\cstabn)$ is reducible into a finite union of projective $(r-1)$-dimensional linear subspaces, one of these subspaces must contain $\nu_n(\proj^1)$. To complete the proof, recall that the affine cone over $\spn(\nu_n(\proj^1))$ is $S^n(\complex^2)$.
\end{proof}

The next proposition shows that $\chi_n^{\reals}$ and $\chi_n$ differ by at most a constant factor of two, which motivates the study of $\chi_n^{\reals}$.
\begin{proposition}\label{real_reduction}
For any positive integer $n$, it holds that $\chi_n \leq \chi^{\reals}_n \leq 2 \chi_n$.
\end{proposition}

\begin{proof}
The first inequality is obvious, so it suffices to prove $\chi^{\reals}_n \leq 2 \chi_n$. Since the tensors $\sum_{\substack{\abs{x}=k}} \mathrm{e}_x$ for $k \in [n]$ form a basis for $S^n(\complex^2)$, then by Proposition~\ref{method}, $\chi_n$ is the minimum number for which there exists a set of stabilizer states $\{\sigma_1,\dots, \sigma_{\chi_n}\} \subseteq \stabn$ such that 
\begin{align}\label{eq:spanny}
\sum_{\substack{x \in \Fn\\\abs{x}=k}} \mathrm{e}_x \in \spn\{\sigma_1,\dots, \sigma_{\chi_n}\}
\end{align}
for all $k \in \{0,1,\dots, n\}$. Without loss of generality, for each $j \in [\chi_n]$ it holds that
\begin{align}
\sigma_j = \frac{1}{\sqrt{\abs{A_j}}} \sum_{x \in A_j} i^{l_j(x)} \cdot (-1)^{q_j(x)} \cdot \mathrm{e}_x
\end{align}
for some affine linear subspace $A_j \subseteq \Fn$, linear functional $l_j$, and quadratic form $q_j$ (see~\eqref{eq:normalformofstabstate}). For each $j \in [\chi_n]$, we define $\rho_j$ to be the state

\begin{align}
\rho_j=\frac{1}{\sqrt{\abs{A_j}}}\sum_{x \in A_j} (-1)^{q_j(x)+l_j(x)} \cdot \mathrm{e}_x,
\end{align}
and observe that for any real numbers $a,b \in \reals$, it holds that
\begin{align}\label{eq:realification}
\setft{Re}((a+ib)\sigma_j)=(a-b) \setft{Re}({\sigma}_j)+b \rho_j,
\end{align}
where $\setft{Re}(\cdot)$ is the real part with respect to the computational basis. For each $k \in \{0,1,\dots, n\}$, there exist real numbers $a_1,\dots, a_{\chi_n}, b_1,\dots, b_{\chi_n}\in \reals$ for which
\begin{align}
\sum_{\substack{\abs{x}=k}} \mathrm{e}_x&=\sum_{j \in [\chi_n]} (a_j+i b_j) \sigma_i\\
								&=\setft{Re}\left(\sum_{j \in [\chi_n]} (a_j+i b_j) \sigma_j \right)\\
								&=\sum_{j \in [\chi_n]} \big((a_j-b_j) \setft{Re}({\sigma}_j) +b \rho_j\big),
\end{align}
where the first line follows from~\eqref{eq:spanny}, the second follows from the fact that the tensor is already real, and the third follows from~\eqref{eq:realification}. Hence,
\begin{align}
S^n(\complex^2)\in \spn_{\complex}\{\setft{Re}({\sigma}_1),\dots, \setft{Re}({\sigma}_{\chi_n}), \rho_1,\dots, \rho_{\chi_n}\}.
\end{align}
Since $\setft{Re}(\sigma_j)$ and $\rho_j$ are clearly proportional to real stabilizer states for all $j \in [\chi_n]$, it follows that ${\chi_n^{\reals} \leq 2 \chi_n}$.
\end{proof}

Using similar techniques as in the proof of Proposition~\ref{method}, it is straightforward to show that $\chi(\psi)\leq \chi^{\reals}(\psi) \leq 2 \chi(\psi)$ for any real quantum state $\psi \in \Spn$.

The next proposition tells us that upper bounds on the stabilizer ranks (respectively, real stabilizer ranks) of any set of $n+1$ pairwise non-collinear states of the form $\psin$ implies an upper bound on $\chi_n$ (resp. $\chi_n^{\reals}$).
\begin{proposition}\label{prop:another_upper}
Let $r$ be a positive integer. If there exists a set of $n+1$ pairwise non-collinear qubit states $\{\psi_1,\dots, \psi_{n+1}\} \subseteq \Sp(\complex^2)$ with $\chi(\psi_i^{\otimes n}) \leq r$ for all $i \in [n+1]$, then $\chi_n \leq r(n+1)$. Similarly, if there exists a set of $n+1$ pairwise non-collinear qubit states $\{\psi_1,\dots, \psi_{n+1}\} \subseteq \Sp(\complex^2)$ with $\chi^{\reals}(\psi_i^{\otimes n}) \leq r$ for all $i \in [n+1]$, then $\chi_n^{\reals} \leq r(n+1)$.
\end{proposition}
\begin{proof}
We prove only the statement for $\chi_n$. First note that $\spn\{\psi_1^{\otimes n}, \dots, \psi_{n+1}^{\otimes n}\}= S^n(\complex^2)$ by \cite[Proposition 3.1]{1751-8121-48-4-045303} or \cite[Corollary 18]{Kruskal_gen}. Since $\chi(\psi_i^{\otimes n}) \leq r$ for all $i \in [n+1]$, then any superposition of the states $\psi_1^{\otimes n},\dots, \psi_{n+1}^{\otimes n}$ has stabilizer rank at most $r(n+1)$. It follows that $\chi_n \leq r(n+1)$.
\end{proof}

We close this section by refining Fact~\ref{fact:generic} to give a quantitative upper bound on the number of states of sub-generic rank.

\begin{proposition}\label{prop:upper_bound}
There are at most
\begin{align}
n \binom{\abs{\cstabn}}{\chi_n-1}
\end{align}
states mod phase $\ppsi \in \proj^1$ for which $\chi(\psin)< \chi_n$. Similarly, there are at most
\begin{align}
n \binom{\bigabs{\cstabn^\reals}}{\chi_n^\reals-1}
\end{align}
states mod phase $\ppsi \in \proj^1$ for which $\chi^{\reals}(\psin)< \chi_n^\reals$.
\end{proposition}
We have computed the quantities $\abs{\cstabn}$ and $\bigabs{\cstabn^\reals}$ at the end of Section~\ref{sec:basicprops}.
\begin{proof}
We prove only the statement for $\chi_n$. Note that any set of $n+1$ pairwise non-collinear states in $\hat{\nu}_1(\proj^1)$ (i.e. states of the form $\psi^{\otimes n}$) span $S^n(\complex^2)$ (see \cite[Proposition 3.1]{1751-8121-48-4-045303} or \cite[Corollary 18]{Kruskal_gen}), so any set of $\chi_n-1$ linearly independent stabilizer states can contain at most $n$ distinct elements of $\hat{\nu}_1(\proj^1)$ in their span. Since there are at most
\begin{align}
\binom{\abs{\cstabn}}{\chi_n-1}
\end{align}
distinct sets of $\chi_n-1$ linearly independent stabilizer states, the bound follows.
\end{proof}

\bibliographystyle{quantum}
\bibliography{stab}
\begin{appendix}
\section{Motivation behind stabilizer rank}\label{app:motivationofstabrank}
In this Appendix we will motivate the definition of stabilizer rank in more depth, and explain how it relates to the simulation cost of quantum circuits. In particular, we will review in detail how the stabilizer rank of a quantum state quantifies the classical simulation cost of applying Clifford gates and computational basis measurements to that state, and review how the stabilizer rank of $n$ copies of the so-called $T$-state quantifies the simulation cost of Clifford+$T$ circuits utilizing $n$ $T$-gates. We use the standard graphical notation for quantum circuits: \replacement{Fixing some natural numbers $m\leq k$, we depict with}

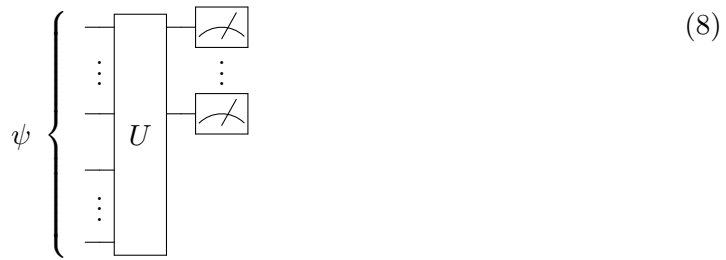
\begin{equation}\label{eq:basiccircuit}
    \Qcircuit @C=.5em @R=.8em {
    & & \qw & \multigate{5}{U} & \qw &\meter\\
    & & \raisebox{4pt}{\vdots} & & &\raisebox{4pt}{\vdots}\\
    & & \qw & \ghost{U} & \qw &\meter\\
    & & \qw & \ghost{U} & &\\
    && \raisebox{4pt}{\vdots} &&&  \\
    & & \qw & \ghost{U} &  & \protect\inputgroupv{1}{6}{1em}{3.8em}{\psi}
    }
\end{equation}
\replacement{a circuit applying a $k$-qubit unitary $U$ to an input state $\psi \in \Sp(\V_k)
$ and measuring an output register consisting of $m$ qubits in computational basis. Without loss of generality, we will always take the first $m$ qubits as output register. In the following, the number of qubits in a circuit as in \eqref{eq:basiccircuit} will always be $k$ and the number of measured qubits $m$ unless specified differently.}

\replacement{Let us start by recalling some basic notions of simulation of quantum circuits. Simulating a quantum circuit as in \eqref{eq:basiccircuit} with a $k$-qubit input state $\psi$ \textit{weakly} means to draw $m$ classical bits according to the output distribution of the circuit. Simulating the circuit \textit{strongly} on the other hand means to be able to compute the probability of a bitstring $x_1\dots x_m$ being the output of the circuit. By Born's rule, this probability is given by}
\begin{equation}\label{eq:overlap}
    P(x_1 \dots x_m) = \psi^* U^\dagger \left( \Pi_{x_1 \dots x_m} \otimes \id_2^{\otimes k-m}  \right) U \ket{\psi}
\end{equation}
\replacement{where $\Pi_{x_1 \dots x_m}$ is the orthogonal projection onto $\spn\{\e_{x_1} \otimes \dots \otimes \e_{x_m}\}$.}

\replacement{Another standard notion of simulation is the so-called $\epsilon$\textit{-strong simulation}. For a bitstring $x_1 \dots x_m$, the task is to approximate the output probability $P(x_1 \dots x_m)$ up to relative error $\epsilon$. More precisely, fix a relative error $\epsilon$.  Given a bitstring $x_1 \dots x_m$, we then want to get an output $\xi$ such that}
\begin{equation*}
    (1-\epsilon ) P(x_1 \dots x_m ) \leq \xi \leq (1 + \epsilon ) P(x_1 \dots x_m )
\end{equation*}
\replacement{where $P$ is the output distribution of the circuit.}
 
Recall the \textit{(affine) Clifford gates} $H, S,$ and $CNOT$ Section \ref{sec:basicprops}. We will later also need the controlled $Z$ gate defined via 
\begin{equation}
    CZ = (\id_2 \otimes H) CNOT (\id_2 \otimes H) = \begin{pmatrix}
    1&0&0&0\\
    0&1&0&0\\
    0&0&1&0\\
    0&0&0&-1\\
    \end{pmatrix}.
\end{equation}
 The Clifford unitaries $U \in \mathrm{Cliff}_k$ are exactly those unitaries composed from $H,S,CNOT,$ and global phase gates only. As done in Section \ref{sec:basicprops}, on can equivalently define the group $\mathrm{Cliff}_k$ as the normalizer of the $k$-qubit Pauli group $\setft{Pauli}_k$ (also defined in Section \ref{sec:basicprops}): Conjugating a Pauli unitary with a Clifford unitary yields another Pauli unitary~\cite{gottesmanphd}.

As a consequence of this fact, we can work very efficiently on certain quantum states using the \textit{stabilizer formalism}, which was introduced in~\cite{gottesmanphd}. We will now briefly recall some basics about the stabilizer formalism and refer the reader to~\cite{nielsen00} for an in-depth discussion and detailed proofs. To start, note that $\e_0^{\otimes k}$ is, up to global phase, the unique quantum state invariant under applying (or \textit{stabilized by}) $Z_i$ for all $i \in [k]$, where $Z_i$ is the tensor product of a Pauli $Z$ on the $i$'th qubit with identities on all other qubits. In fact, one can show that if $P_1, \dots , P_k$ are independent Pauli unitaries (that is, none of them is a product of the others), and the group generated by them does not contain $-\id_2^{\otimes k}$, then there exists up to global phase a unique state $\sigma$ stabilized by $P_1, \dots , P_k$ (see~\cite[Chapter 10.5.1]{nielsen00} for details).

It turns out that if a state $\sigma$ can be uniquely specified as above, the same holds for the state resulting from applying a Clifford unitary $U$ to $\sigma$: In fact, one can show that the Pauli unitaries $UP_1 U^\dagger , \dots , UP_k U^\dagger$ are independent again, do not generate $-\id_2^{\otimes k}$ and it is easy to see that they stabilize $U\sigma$ (see~\cite[Chapter 10.5.2]{nielsen00} for details). Conversely, for every state of the form $U\e_0^{\otimes k}$, there are $k$ independent Paulis stabilizing it, see~\cite{gross2006hudson} for a proof. Summarizing, we see that the states of the form $\sigma = U\e_0^{\otimes k}$ for a Clifford unitary $U$ are exactly the ones uniquely specified up to global phase by $k$ independent Pauli operators stabilizing it. This explains why we defined a \textit{stabilizer state} as a state of the form $U\e_0^{\otimes k}$ where $U$ is a Clifford unitary in Section~\ref{sec:basicprops}.

The celebrated Gottesman-Knill theorem states that a circuit as in \eqref{eq:basiccircuit} where $U$ is a Clifford unitary can be simulated efficiently if the input state is a stabilizer state  ${\sigma \in \Sp(\V_k)}$.

\begin{theorem}[Gottesman, Knill~\cite{Gottesman:1998hu}]\label{thm:gotknilthm}
Let $\sigma$ be a $k$-qubit stabilizer state specified by $k$ independent Paulis that stabilize it. A quantum circuit $U$ composed only from Clifford gates acting on $\sigma$ followed by measuring an output register consisting of $m \leq k$ qubits in the computational basis can be efficiently simulated both strongly and weakly. More precisely, the complexity of simulating the quantum computation scales quadratically in the number of qubits and linearly in the number of Clifford gates and measurements that are applied.
\end{theorem}

A proof of Theorem~\ref{thm:gotknilthm} can be found in~\cite[Section 10.5.4]{nielsen00}. We mention that one essentially has to update the Pauli stabilizers gate-by-gate, which yields $k$ independent Paulis stabilizing $U\sigma$. One can read in detail in~\cite[Section 10.5.3]{nielsen00} how one can then efficiently obtain the outcome probabilities from this. \replacement{It is also worth mentioning that the theorem holds more generally if we measure any observable from the Pauli group, not only for measuring in the computational basis.}

Note that if a quantum state $\sigma \in \Sp(\V_k)$ is stabilized by a Pauli unitary $P$, all states of the form $e^{i \theta} \sigma$ for $\theta \in \mathbb{R}$ are also stabilized by $P$. While a global phase does not influence the outcome probabilities of computational basis measurements, it changes the amplitudes of the state. This means that Theorem~\ref{thm:gotknilthm} as stated above does not let us calculate amplitudes but only outcome probabilities. In~\cite{Bravyilowrankstabdecpmps}, the authors describe a way of keeping track of the global phase: They define a classical data format for stabilizer states which they call the \textit{$CH$-form}. We now briefly describe their approach and refer to~\cite{Bravyilowrankstabdecpmps} for more details. Another in-depth discussion can be found in~\cite[Section 2.1.4]{hammam} from where we also borrow our notation. Essentially, one can show that every stabilizer state $\sigma$ is of the form
\begin{equation}
   \sigma \replacement{=\sigma(w,U,h,s)} =  wUH(h) (\e_{s_1} \otimes \dots \otimes \e_{s_k})
\end{equation}
where $w \in \mathbb{C}$, $s, h \in \mathbb{F}_2^k$ and $U$ is a Clifford unitary composed from gates $CNOT, CZ$ and $S$ only. The unitary $H(h)$ is a tensor product of $H$ gates acting on the qubits $i$ such that $h_i = 1$ and $\id_2$ elsewhere. As in the proof of Theorem~\ref{thm:gotknilthm},

 one can now simulate a circuit built from Clifford gates only by updating $w,U, h$ and $s$ gate-by-gate. \replacement{The authors show in~\cite{Bravyilowrankstabdecpmps} that this update can be done with computational cost at most quadratic in $k$: Updating the CH form after applying a global phase has constant computational cost, a $CNOT$, $CZ$, or $S$ has linear cost in $k$, and a Hadamard gate has quadratic cost in $k$. Moreover, they show that given a stabilizer state $\sigma\in \Sp(\V_k)$ in $CH$-form, one can calculate the amplitude $(\e_{x_1}^* \otimes \dots \otimes \e_{x_k}^*) \sigma$ with computational cost quadratic in the number of qubits. Consequently, one can calculate the overlap}
 \begin{equation*}
     \sigma(w,U,h,s)^*\sigma(w',U',h',s') = \overline{w}\cdot  w'\cdot  (\e_{s_1}^* \otimes \dots \otimes \e_{s_k}^*) H(h) U^\dagger \sigma(w',U',h',s')
 \end{equation*}
\replacement{of two stabilizer states given in $CH$-form with computational cost at most cubic in the number of qubits by first updating the $CH$-form of $H(h) U^\dagger \sigma(w',U',h',s')$ and then calculating the overlap with $(\e_{s_1} \otimes \dots \otimes \e_{s_k})$. Finally, for a stabilizer state $\sigma(w,U,h,s)$ given in $CH$-form, an integer $m \leq k$, and an orthogonal projector $\Pi_{x_1 \dots x_m}$ onto ${\spn\{\e_{x_1}\otimes \dots \otimes \e_{x_m}\}}$ for some bitstring $x_1\dots x_m$, one can calculate the $CH$-form of the stabilizer state }
 \begin{equation*}
     \left(\Pi_{x_1 \dots x_m} \otimes \id_2^{\otimes k-m} \right) \sigma(w,U,h,s)
 \end{equation*}
\replacement{with compuational cost quadratic in $k$.}

With that, we can now understand the importance of the stabilizer rank as a measure of the computational cost of strong simulation of quantum circuits with general input states. Say, we want to strongly simulate a quantum circuit as in \eqref{eq:basiccircuit} where $\psi$ is a $k$-qubit quantum state and $U$ is a Clifford unitary. Say furthermore that we can decompose 
\begin{equation}\label{eq:stabdecomp}
    \psi = \sum_{i = 1}^r c_i \sigma_i
\end{equation}
where the $\sigma_i$ are stabilizer states each specified in $CH$-form and the $c_i$ are complex numbers. \replacement{Given a bitstring $x_1 \dots x_m$, we want to calculate} 
\begin{align*}
   P(x_1 \dots x_m) = \psi^* U^\dagger \left( \Pi_{x_1 \dots x_m} \otimes \id \right) U \ket{\psi} = \sum_{i,j = 1}^r \overline{c_i}\cdot c_j\cdot \sigma_i^* U^\dagger \left(\Pi_{x_1 \dots x_m} \otimes \id_2^{\otimes k-m}\right) U \sigma_j.
\end{align*}
\replacement{We can do so by first updating the $CH$-form of $U\sigma_i$ and $(\Pi_{x_1 \dots x_m} \otimes \id_2^{\otimes k-m}) U \sigma_j$ followed by calculating $r^2$ overlaps between stabilizer states. Since for every of the summands the computational cost is linear in the number of gates of which $U$ is composed and cubic in the number of qubits, the whole cost is cubic in the number of qubits, linear in the number of Clifford gates and quadratic in the number of terms $r$ appearing in the stabilizer decomposition.}

As already done in Section \ref{sec:basicprops}, we call the minimal $r$ such that there exists a decomposition as in equation \eqref{eq:stabdecomp} the \textit{stabilizer rank} of $\psi$ and denote it $\chi(\psi)$. It follows that lower bounds on $\chi(\psi)$ imply lower bounds on the complexity of simulating a circuit of the form \eqref{eq:basiccircuit} using the approach we have just described. On the other hand, finding upper bounds on $\chi(\psi)$ by providing a decomposition as in \eqref{eq:stabdecomp} can give us better simulation algorithms for circuits of the form \eqref{eq:basiccircuit}.

\replacement{For a discussion of a weak simulation protocol and an $\epsilon$-strong simulation protocol making use of stabilizer rank decompositions see \cite[Section 4]{Bravyilowrankstabdecpmps}. We mention that there the authors deduce that the computational cost of weak simulation scales linearly in the approximate stabilizer rank of $\psi$ and polynomially in the number of qubits and the number of gates applied; and that the cost of $\epsilon$-strong simulation scales linearly in the exact stabilizer rank of $\psi$ and polynomially in the number of qubits and the number of gates applied.}

As an application of the connection between the stabilizer rank of a state $\psi$ and the classical simulation cost of applying Clifford circuits and computational basis measurements to $\psi$, we will now see how low-rank stabilizer decompositions of $n$ copies of the so-called $T$-state can also be used to simulate circuits built from a universal gate set. This strategy has been used for instance in~\cite{BravyiImprovedSimDominated}. \replacement{For simplicity, we will only consider the case here where all qubits in the circuit in \eqref{eq:basiccircuit} are measured in computational basis, that is, the size $m$ of the output register is equal to the number of qubits $k$.}

Recall that a gate set $\mathcal{G}$ is called \textit{universal} if every unitary $U$ on $k$ qubits can be approximated arbitrarily well by unitaries $U_\mathcal{G}$ composed solely of gates in $\mathcal{G}$. It is well-known that the set of Clifford gates is not universal but, together with the so-called $T$-gate 

\begin{equation}\label{eq:Tgate}
T = \begin{pmatrix}1&0\\0&e^{i\frac{\pi}{4}}\end{pmatrix},
\end{equation}
they form a universal gate set~\cite{Boykin99onuniversal}.

Under the assumption that our quantum device can only prepare a computational basis state $\e_0^{\otimes k}$, apply Clifford$+T$ operations and measure the qubits in computational basis, the $T$-gates appear, by the preceeding discussion, to be responsible for the potential superiority of the quantum device. It is therefore an interesting question how efficiently we can simulate a quantum computation on $k$ qubits which, in addition to Clifford gates, uses $n$ single-qubit $T$-gates.

A standard way of approaching this question is via the study of \textit{magic state injection}~\cite{PhysRevA.71.022316}. For this, we note that for the magic state $T = \frac{1}{\sqrt{2}}\left(\mathrm{e}_0 + e^{i\frac{\pi}{4}} \mathrm{e}_1\right)$, applying a $T$-gate to any qubit state $\psi$ is the same as applying the circuit

\begin{equation}\label{eq:injection}
    \Qcircuit @C=.5em @R=0.7em  {
      &\lstick{\psi}&\qw&\qw&\qw&\qw&\qw&\qw &\qw &\ctrl{1} &\qw &\qw &\qw &\gate{S} &\qw &\qw \\
    &&&&&&&&\lstick{T}\qw&\targ&\qw &\meter &\cw &\controlo \cw \cwx \gategroup{1}{5}{2}{14}{1.3em}{--}&& 
    }
\end{equation}
where the double wire denotes a \textit{classical control}: The $S$ gate on the first system is applied if and only if the outcome of the computational basis measurement on the second system is 1. Therefore, any $k$-qubit quantum circuit

\begin{equation}\label{eq:basiccircuitwitht}
    \Qcircuit @C=.5em @R=.8em {
    & & \qw & \multigate{2}{V} & \qw &\meter\\
    &\lstick{\e_0^{\otimes k\;}\Bigg\lbrace} & \raisebox{4pt}{\vdots} & & &\raisebox{4pt}{\vdots}\\
    & & \qw & \ghost{V} & \qw &\meter
    }
\end{equation}
where $V$ is composed of Clifford gates and $n$ $T$-gates, can be implemented with a circuit composed from Clifford gates and classical controls acting on the state $\mathrm{e}_0^{\otimes k} \otimes T^{\otimes n}$ by replacing each $T$-gate with the gadget in \eqref{eq:injection}. By \textit{postselecting} outcomes $0$ for each measurement on a $T$-state in this circuit, we see that 
\begin{equation}\label{eq:postselect}
    V \e_0^{\otimes k} = {2^{n/2}}\left(\id^{\otimes k} \otimes (\e_0^*)^{\otimes n}\right) U \left(e_0^{\otimes k} \otimes T^{\otimes n}\right) 
\end{equation}
where $U$ is composed from Clifford gates only (more precisely, the Clifford gates from $V$ plus an additional $CNOT$ gate for each injected $T$-state).

With that, it follows that the probability of outcome $x_1 \dots x_k$ in the circuit in \eqref{eq:basiccircuitwitht} is given by $\abs{p_{x_1\dots x_k}}^2$, where 

\begin{align}
  p_{x_1 \dots x_k} &=\left(\e_{x_1}^* \otimes \dots \otimes \e_{x_k}^*\right) V \e_0^{\otimes k}\\ &={2^{n/2}}\left(\e_{x_1}^* \otimes \dots \otimes \e_{x_k}^* \otimes (\e_0^*)^{\otimes n}  \right)U \left((\e_0)^{\otimes k} \otimes T^{\otimes n}\right).\label{eq:probabilitywitht}
\end{align}
Note that if there are stabilizer states $\sigma_1,\dots, \sigma_r$ and complex numbers $c_1,\dots, c_r \in \complex$ such that $\psi=\sum_{i=1}^r c_i \sigma_i$, then the quantity in \eqref{eq:probabilitywitht} decomposes by linearity as 

\begin{equation}\label{eq:probabilitywithstb}
  p_{x_1 \dots x_k} = 2^{n/2} \sum_{i=1}^r c_i \left(\e_{x_1}^* \otimes \dots \otimes \e_{x_k}^* \otimes (\e_0^*)^{\otimes n}  \right)U \left((\e_0)^{\otimes k} \otimes \sigma_i\right).
\end{equation}
All summands in \eqref{eq:probabilitywithstb} can be calculated efficiently using the $CH$-form as discussed before: They are amplitudes of the outcome of a quantum circuit composed of Clifford gates acting on a stabilizer state. 

Summarizing, every stabilizer decomposition of $T^{\otimes n}$ yields a way to strongly simulate Clifford+$T$ circuits consisting of $n$ $T$-gates. The complexity of this simulation scales linearly in the number of terms appearing in the stabilizer decomposition and polynomially in all other parameters. With this, it follows that finding decompositions of $T^{\otimes n}$ into few stabilizer states can reduce the complexity of simulating such circuits. Lower bounds on $\chi(T^{\otimes n})$, on the other hand, translate directly into lower bounds on the cost of simulating quantum circuits using these methods, namely, stabilizer decompositions and $CH$-forms.

\end{appendix}

\end{document}